\documentclass[final,12pt]{econsocart}

\usepackage{booktabs}
\usepackage[yyyymmdd,24hr]{datetime}
\usepackage{caption}

\usepackage{mathtools}
\RequirePackage[colorlinks,citecolor=blue,linkcolor=blue,urlcolor=blue,pagebackref]{hyperref}

\startlocaldefs
\usepackage{enumerate}

\theoremstyle{plain}

\newtheorem{theorem}{Theorem}
\newtheorem{proposition}{Proposition}

\newtheorem{lemma}{Lemma}

\theoremstyle{remark}

\usepackage{tikz}
\usetikzlibrary{shapes}
\usetikzlibrary{arrows.meta, positioning,calc}
\usetikzlibrary{patterns, graphs}

\newcommand\encircle[1]{%
	\tikz[baseline=(X.base)] 
	\node (X) [draw, shape=circle, inner sep=0] {\strut #1};}

\newcommand\ensquare[1]{%
	\tikz[baseline=(X.base)] 
	\node (X) [draw, shape=rectangle, inner sep=.1cm] {\strut #1};}

\usepackage{etoolbox}

\newcommand{\poi}{\textup{Rw}} 
\newcommand{\poik}{\textup{RkI}}

\newcommand{\da}{\mathtt{DA}}

\newcommand{\rmm}{\mathtt{RM}} 
\newcommand{\rw}{\mathtt{RW}}

\newcommand{\rk}{\textup{rk}}

\usepackage{setspace}

\usepackage{enumitem}
\usepackage{amstext} 
\usepackage{array}   
\newcolumntype{C}{>{$}c<{$}} 
\usepackage{multirow}
\usepackage{float}
\usepackage{array}   
\newcolumntype{C}{>{$}c<{$}} 

\endlocaldefs

\begin{document}
	\begin{frontmatter}
		
		\title{What Pareto-Efficiency Adjustments Cannot Fix}

				\begin{aug}
			\author[id=au1,addressref={add1}]{\fnms{Josu\'e}~\snm{Ortega}}
			\author[id=au2,addressref={add2}]{\fnms{Gabriel}~\snm{Ziegler}}
			\author[id=au3,addressref={add3}]{\fnms{R. Pablo}~\snm{Arribillaga}}
			\author[id=au4,addressref={add4}]{\fnms{Geng}~\snm{Zhao}}
			
			\address[id=add1]{\orgname{Queen's University Belfast}}
			
			\address[id=add2]{\orgname{University of Edinburgh}}
			
			\address[id=add3]{\orgdiv{Instituto de Matemática Aplicada San Luis},
				\orgname{Universidad Nacional de San Luis, and CONICET}}
			
			\address[id=add4]{\orgname{University of California, Berkeley}}
		\end{aug}
		
		\support{Emails: \href{mailto:j.ortega@qub.ac.uk}{j.ortega@qub.ac.uk}, \href{mailto:ziegler@ed.ac.uk}{ziegler@ed.ac.uk}, \href{mailto:rarribi@unsl.edu.ar}{rarribi@unsl.edu.ar},  \href{mailto:gengzhao@berkeley.edu}{gengzhao@berkeley.edu}.}


\begin{abstract}
The Deferred Acceptance (DA) algorithm is stable and strategy-proof, but can produce outcomes that are Pareto-inefficient for students, and thus several alternative mechanisms have been proposed to correct this inefficiency. However, we show that these mechanisms cannot correct DA's rank-inefficiency and inequality, because these shortcomings can arise even in cases where DA is Pareto-efficient.

We also examine students' segregation in settings with advantaged and marginalized students. We prove that the demographic composition of every school is perfectly preserved under any Pareto-efficient mechanism that dominates DA, and consequently fully segregated schools under DA maintain their extreme homogeneity. \\
\end{abstract}

		\begin{keyword}
			\kwd{school choice}
			\kwd{inequality}
			\kwd{rank-efficiency}
			\kwd{segregation}
		\end{keyword}
		
		\begin{keyword}[class=JEL] 
			\kwd{C78}
			\kwd{D47}
		\end{keyword}

\end{frontmatter}

\onehalfspacing	
\newpage
\setcounter{secnumdepth}{3} 

\section{Introduction}
\label{sec:introduction}
School choice mechanisms shape the educational opportunities of millions of students worldwide. Among these, the student-proposing Deferred Acceptance (DA) algorithm has emerged as the gold standard, adopted by major cities from New York to Paris. Its appeal lies in two compelling properties: it produces stable matchings that respects schools' priorities while incentivizing parents to truthfully report their preferences. Yet these virtues come at a cost---DA often generates Pareto-inefficient outcomes, leaving some students worse off than necessary. This has been observed in theoretical and empirical studies alike \citep[e.g.][]{abdulkadirouglu2003,abdulkadirouglu2009strategy,kesten2010school}.

DA's Pareto inefficiency has sparked a growing literature in market design, leading to the construction of sophisticated mechanisms that improve students' welfare compared to DA's baseline while preserving weaker forms of stability and strategy-proofness. The point of this paper is that such efficiency-adjusted mechanisms may not correct other problems of DA-based allocations that have received less attention in the literature: namely their rank-inefficiency, inequality, and segregation. This concern is underscored by the observation that Pareto-efficiency may be \emph{``a very limited kind of success, and in itself may or may not guarantee much''} \citep[p.31--32]{sen1999ethics}.

\paragraph{A Motivating Example.}To illustrate the problems of inequality and rank-inefficiency, consider the school choice problem in Table \ref{tab:example1} below with six students ($i_1,\ldots,i_6$) and six schools ($s_1, \ldots, s_6$) with unit capacity each.
	
	\begin{table}[h!]
		\centering
		\caption{\centering DA allocation in squares.}
		\label{tab:example1}
		\begin{tabular}{CCCCCC|CCCCCC}
			\toprule
			i_1 & i_2 & i_3 & i_4 & i_5 & i_6 & s_1 & s_2 & s_3 & s_4 & s_5 & s_6 \\
			\midrule
			\ensquare{$s_1$} 	 & \encircle{$s_1$} & \encircle{$s_2$} & \encircle{$s_3$} & \encircle{$s_4$} & \encircle{$s_5$} & i_1 & i_2 & i_3 & i_4 & i_5 & i_6  \\
			\encircle{$s_6$} 	 & \ensquare{$s_2$} & s_1 & s_1 & s_1 & s_1 & i_2 & i_3 & i_4 & i_5 & i_6 & i_1   \\
			&    & \ensquare{$s_3$} & s_2 & s_2 & s_2 & i_3 & i_4 & i_5 & i_6 &    &    \\
			&    &    & \ensquare{$s_4$} & s_3 & s_3 & i_4 & i_5 & i_6 &    &    &    \\
			&    &    &    & \ensquare{$s_5$} & s_4 & i_5 & i_6 &    &    &    &    \\
			&    &    &    &    & \ensquare{$s_6$} & i_6 &    &    &    &    &   
		\end{tabular}
	\end{table}

	The unique stable allocation  appears in squares, with every student assigned to the school with the same index. Three students are assigned to schools in the bottom half of their preferences, with student $i_6$ assigned to his least-preferred school. Note that this allocation is already Pareto-efficient.
	Interestingly, an alternative Pareto-efficient allocation, in circles, is remarkably more compelling because of two important reasons:
	
	\begin{enumerate}
		\item \emph{it is more egalitarian}: every student is assigned to either the first or second ranked school, unlike in DA where some student ends up in his worst-ranked school,
		
		\item \emph{it is more rank-efficient}: the median student gets assigned to his first choice, rather than to a school in the middle of his preference list, with five out of six students preferring the allocation in circles.\\		
	\end{enumerate}

Interestingly, because DA's allocation was Pareto-efficient in our previous example, any Pareto-efficient mechanism that weakly Pareto dominates DA---to which we henceforth refer as \emph{stable-dominating} \citep{alva2019stable}---also selects the same allocation.\footnote{Alva and Manjunath only require that a stable-dominating matching weakly Pareto-dominates some stable matching, not necessarily the student-optimal one. We impose the stronger restriction motivated by our focus on mechanisms that improve on student-proposing DA, which is the status quo in many school districts.}
Thus, the previous example illustrates a concern of stable and stable-dominating mechanisms: they may select highly unequal and rank-inefficient Pareto-efficient allocation when alternative Pareto-efficient allocations are more appealing.\footnote{This flaw is not exclusive of stable-dominating mechanisms: the top trading cycle (TTC) mechanism is neither stable nor stable-dominating and also selects the rank-inefficient and unequal matching in squares.} In Proposition \ref{prop:rawlsian}, we generalize the previous example to provide tight bounds quantifying the largest inequality and rank-inefficiency generated by stable and stable-dominating mechanisms.

\paragraph{The Persistency of Segregation.}A related concern is DA's tendency to produce segregated allocations, where marginalized students are concentrated in schools with predominantly low-achieving peers (as in Chile \citep{chileseg} and the UK \citep{terrier2021immediate}). This observation raises a second  research question: {can the efficiency improvements from any stable-dominating mechanism foster greater school diversity and create more opportunities for marginalized students than those under DA?}

Our analysis in Theorem \ref{thm:preservation} reveals fundamental limitations of stable-dominating mechanisms for reducing segregation. In settings where students can be classified as either advantaged or marginalized---with marginalized students having lower priority at every school---we prove that the demographic composition of every school is perfectly preserved under any stable-dominating mechanism. Students can only trade positions within their own group, never across the advantaged-marginalized divide. This means that any school that is fully segregated under DA---admitting exclusively either advantaged or marginalized students---remains segregated under any Pareto-efficient mechanism that dominates DA. More strikingly, even mixed schools maintain their exact proportion of advantaged and marginalized students, as advantaged students at such schools cannot be displaced through any efficiency-improving trades.

\section{Related Literature}
\label{sec:literature}

\paragraph{Stable-Dominating Mechanisms.} A growing literature has developed mechanisms that Pareto-dominate DA. The Efficiency-Adjusted Deferred Acceptance (EADA) mechanism \citep{kesten2010school} stands as the most prominent example, allowing controlled priority violations to achieve Pareto improvements. EADA satisfies several refined stability concepts including sticky stability \citep{afacan2017sticky}, essentially stability \citep{troyan2020essentially}, and weak-stability \citep{tang2021weak}, as well as weak strategy-proof properties such as non-obvious manipulability \citep{troyan2020obvious} and regret-free truth-telling \citep{regret}, demonstrating the theoretical sophistication possible within this framework.

Beyond EADA, this class includes the Top Trading Cycles mechanism using DA's allocation as endowment, the Maximum Improvement over DA (MIDA) mechanism \citep{knipe2025improvable}, and various approaches within the Student Exchange under Partial Fairness (SEPF) framework \citep{dur2019school}. The practical appeal of these mechanisms is evidenced by EADA's scheduled implementation in Flanders \citep{cerrone2022school}.

Our analysis examines whether these Pareto improvements can also address other distributional concerns about DA allocations, particularly regarding rank-efficiency, inequality, and segregation.

\paragraph{Rank-Efficiency.} The average rank under DA significantly exceeds (i.e. is worse) that obtained by the rank-minimizing mechanism, and this occurs in both random markets and a dataset from the Hungarian high school allocation \citep{ortega2023cost}. However, DA's average rank is roughly the same as that obtained by other Pareto-efficient mechanisms such as TTC and random serial dictatorship. \cite{pycia2019evaluating} provides a detailed overview of the literature and a theoretical explanation. Regarding stable-dominating mechanisms, \cite{diebold2017matching} and \cite{ortega2023cost} have examined the rank-efficiency on stable-dominating mechanisms in practice, that while better than DA's, remains far from that attained by the first best.

\paragraph{Inequality.} DA can produce highly unequal allocations, with the worst-off student receiving significantly worse placements than under Pareto-efficient alternatives \citep{galichon2023stable}. Stable-dominating mechanisms cannot improve the placement of unassigned or poorly matched students, and thus are limited in the extent to which they can reduce inequality \citep{tang2014new,alva2019stable,ortega2025}. Efficiency-adjustments over DA such as EADA can increase variance in the rank distribution, yet reducing the corresponding	 Gini index \citep{cerrone2022school}.
Nevertheless, our companion paper shows that almost all students could trade their DA placement to be assigned to a more desirable school in large random markets, so that even when inequality does not reduce, the welfare of most students improves under any efficient stable-dominating mechanism \citep{ortega2025}.

\paragraph{Segregation.} Empirical evidence from England and Chile shows DA can worsen segregation: in England, switching to DA harmed low-SES students who attended schools with lower value-added and more disadvantaged peers \citep{terrier2021immediate}; in Chile, DA increased within-school segregation in areas with high residential segregation or many private schools \citep{chileseg}. More broadly, school choice policies can significantly increase segregation, with every 3\% increase in choice availability affecting over 564,000 US children \citep{ukanwa2022school}. However, alternative mechanisms using minority reserves can reduce segregation without harming efficiency: simulations of Swedish school choice show that reserved seats reduce segregation while maintaining high shares of students assigned to preferred schools \citep{kessel2018school}, while in Chile such reserves would reduce segregation while improving rank-efficiency \citep{escobar2022segregation}. Theoretically, \cite{calsamiglia2023catchment} show how DA can limit disadvantaged students' access to good schools, albeit it being immune to a type of segregation that arises due to heterogeneous risk attitudes \citep{calsamiglia2021school}. 

\paragraph{Our Contribution.} While the literature establishes that DA suffers from inequality, rank-inefficiency, and segregation problems, our contribution is proving that stable-dominating mechanisms cannot eliminate these shortcomings. This highlights limitations of Pareto-efficiency reforms in addressing distributional concerns.

\section{Model}
\label{sec:model}

Following \cite{abdulkadirouglu2003}, a school choice problem $P$ consists of a finite set of students $I$ and a finite set of schools $S$. Each student $i$ has a strict preference $\succ_i$ over the schools. Each school $s$ has a quota of available seats $q_s$ and a strict priority over the students $\triangleright_s$, determined by local educational regulations. To allow for unassigned students, we allow the existence of a null school (denoted $s_\emptyset$), which has unlimited capacity. 
We use $m=|I|$ and $n=|S|$ to denote the number of students and schools, respectively. 
Let $\mathcal{P}_{m,n}$ denote the set of all school choice problems with $m$ students and $n$ schools.

For a given school choice problem $P$, a \textit{matching} $\mu$ is a mapping from $I$ to $S$ such that no school is matched to more students than its quota. We denote by $\mu_i$ the school to which student $i$ is assigned and by $\mu^{-1}_s$ the set of students assigned to school $s$. With this notation, a matching needs to satisfy $|\mu^{-1}_s|\leq q_s$ for every $s \in S$. We call every student $i$ with $\mu_i=s_\emptyset$ \emph{unassigned}. 

The function $\rk_i:S \rightarrow \{1, \ldots, n\}$ specifies the rank of school $s$ according to the preference profile $\succ_i$ of student $i$, i.e.
\begin{align*}
\rk_i(s) = |\{s' \in S: s' \succ_i s\}|+1,
\end{align*}
so that the most desirable option gets a rank of 1, whereas the least desirable gets a rank of $n$. With some abuse of notation, we use the same rank function to specify the students' rank as per the priority profile of schools, i.e. $\rk_s(i)= |\{j \in I: i \triangleright_s j\}|+1$.

A matching $\mu$ \textit{weakly Pareto-dominates} matching $\nu$ if, for every student $i \in I$, $\rk_i(\mu_i) \leq \rk_i (\nu_i)$. A matching $\mu$ \textit{Pareto-dominates} matching $\nu$ if $\mu$ weakly dominates $\nu$ and there exists a student $j \in I$ with $\rk_j(\mu_j) < \rk_j (\nu_j)$. A matching is \textit{Pareto-dominated} if there exists a matching that Pareto-dominates it and is \textit{Pareto-efficient} if it is not Pareto-dominated. 

A strictly stronger efficiency notion is that of rank-efficiency.
A matching $\mu$ is \emph{rank-efficient} if there does not exist an alternative matching $\nu$ such that $\sum_{i \in I} \rk_i(\nu_i) < \sum_{i \in I} \rk_i (\mu_i)$. Every rank-efficient matching is Pareto-efficient but the converse is not true. 

Student $i$ \emph{desires} school $s$ in matching $\mu$ if $\rk_i(s)< \rk_i(\mu_i)$. We say that student $j$ \emph{violates} student $i$'s priority at school $s$ in matching $\mu$ if $i$ desires $s$, $\mu_j =s$, and $\rk_s(i) < \rk_s(j)$. A matching $\mu$ is \emph{non-wasteful} if every school $s$ that is desired by some student in $\mu$ satisfies $|\mu^{-1}_s| =q_s$. A matching $\mu$ is \textit{stable} if it is non-wasteful and no student’s priority at any school is violated in $\mu$.

A \emph{mechanism} associates a matching to every school choice problem. For mechanism $M$, $M(P)$ denotes the resulting matching for school choice problem $P$. We are mainly interested in the \emph{student-proposing Deferred Acceptance (DA) mechanism} \citep{gale1962}, which works as follows:

\begin{enumerate}[leftmargin=2cm]
	\item[Round 1:] Every student applies to her most preferred school. Every school tentatively accepts  the best students according to its priority, up to its capacity, and rejects the rest. 
	\item[Round $k$:] Every student rejected in the previous round applies to her next best school. Among both new applicants and previously accepted students, every school accepts the best students according to its priority, up to its capacity, and rejects the rest.
\end{enumerate}

The procedure stops when there is a round without any new rejection. We use $\da(P)$ to denote the (unique) resulting matching generated by DA in school choice problem $P$.  $\da_i(P)$ denotes the school to which student $i$ is assigned. $\da^{-1}_s(P)$ denotes the set of students assigned to school $s$ in $\da(P)$.

A matching $\nu$ is stable-dominating if, for each student $i \in I$, $\rk_i(\nu_i)\leq \rk_i (\da_i(P))$. Note that the DA is stable-dominating. Stable-dominating matchings were first studied by \cite{alva2019stable}.
A mechanism $M$ is \emph{stable-dominating} if, for every school choice $P$, $M(P)$ is a stable-dominating matching. Stable-dominating mechanisms include EADA \citep{kesten2010school}, TTC using DA's allocation as endowment, and the Maximum Improvement over DA (MIDA) mechanism \citep{knipe2025improvable}. 

\subsection{Preliminary, Known Results}

A student is \emph{unimprovable} if he is assigned to the same school in every stable-dominating mechanism. The following lemmas will be useful for our analysis.

\begin{lemma}[\cite{tang2014new,alva2019stable}]\label{lemma:unim}
	If a student $i$ is either unassigned in DA or assigned to a school that did not reject any student during DA's execution, then student $i$ is unimprovable.
\end{lemma}

A second known lemma that is useful for us states that the number of students assigned to a school is constant across stable-dominating matchings.

\begin{lemma}[\cite{tang2014new, alva2019stable}]\label{lemma:lema3}
	For any school choice problem $P$, any stable-dominating mechanism $M$ and any school $s \in S$, $\left\vert\da^{-1}_s(P)\right\vert=\left\vert M^{-1}_s(P)\right\vert$.
\end{lemma}

The third lemma connects unimprovable students with cycles in the envy digraph induced by the Deferred Acceptance (DA) mechanism—a novel concept introduced in our companion paper \citep{ortega2025}. In this digraph, denoted $G^\da(P)$, each student is a node with directed edges to other students whose assignments they prefer to their own. The envy digraph for our motivating example appears below; note the absence of cycles, which reflects DA’s Pareto-efficiency in this case.

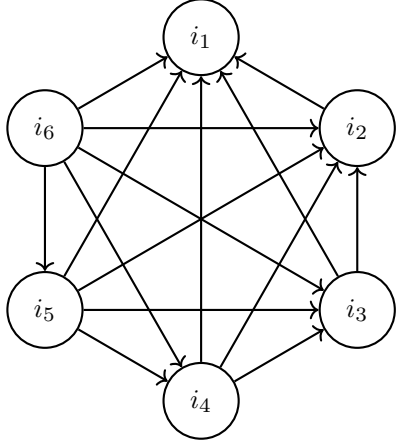
\begin{center}
	\begin{tikzpicture}[scale=0.8, ->, node distance=.5cm, every node/.style={draw, circle, minimum size=1cm}, thick]
		\foreach \i in {1,2,...,6} {
			\node (N\i) at ({90-60*(\i-1)}:3) {$i_\i$};
		}
		
		
		
		\draw[thick] (N2) -- (N1);
		\draw[thick] (N3) -- (N1);
		\draw[thick] (N3) -- (N2);
		\draw[thick] (N4) -- (N1);
		\draw[thick] (N4) -- (N2);
		\draw[thick] (N4) -- (N3);
		\draw[thick] (N5) -- (N1);
		\draw[thick] (N5) -- (N2);
		\draw[thick] (N5) -- (N3);
		\draw[thick] (N5) -- (N4);
		\draw[thick] (N6) -- (N1);
		\draw[thick] (N6) -- (N2);
		\draw[thick] (N6) -- (N3);
		\draw[thick] (N6) -- (N4);
		\draw[thick] (N6) -- (N5);		
	\end{tikzpicture}
\end{center}

The next lemma relates DA's envy digraph and unimprovable students.

\begin{lemma}[\cite{kesten2010school,erdil2014strategy,tang2014new,ortega2025}]
		\label{lemma:cycles}
	A student is unimprovable if and only if he does not belong to any cycles in $G^\da(P)$.
\end{lemma}

\section{Results}
\label{sec:thegap}

\subsection{Bounds}

\paragraph{Rawlsian Inequality.} We assess a mechanism's inequality by comparing the maximum rank assigned to any student (i.e., the rank of the worst-off student) against the minimum possible maximum rank achievable. This minimum is attained by the \emph{Rawlsian mechanism} ($\rw$), which selects a matching that minimizes the maximum rank for every school choice problem $P$ \citep{demeulemeester2022rawlsian,afacan2024rawlsian,kuvalekar2024fair}. This comparison reflects a societal welfare perspective focused on the worst-off individual, following \cite{rawls1971}. Although the Rawlsian matching $\rw(P)$ is not always unique, this potential multiplicity and its resolution are irrelevant for our analysis.

We define the \emph{Rawlsian inequality} (\poi) of a mechanism $M$ for a school choice problem with $n$ schools as the following ratio:\footnote{This concept is related to the \emph{price of anarchy} and \emph{price of stability} concepts in algorithmic game theory \citep{koutsoupias1999worst, anshelevich2008price}. A similar concept was proposed by \cite{boudreau2013preferences} to measure the efficiency loss of stable matching mechanisms.}
\begin{equation}
	\poi(M;n) := \sup_{m \in \mathbb{N}, P \in \mathcal{P}_{m,n}}\left[ \frac{ \max_{i \in I} \rk_i [M_i(P)]}{\max_{i \in I} \rk_i[\rw_i(P)]} \right]
\end{equation}

\paragraph{Rank-Inefficiency.} Analogously, the Rank-Minimizing mechanism ($\rmm$) selects a matching minimizing the average rank for every $P$ \citep{featherstone2020rank}. We define the rank-inefficiency ($\poik$) of $M$ as:
\begin{equation}
	\poik(M;n) := \sup_{m \in \mathbb{N},\ P \in \mathcal{P}_{m,n}} \left[ \frac{ \sum_{i \in I} \rk_i [M_i(P)] }{ \sum_{i \in I} \rk_i[\rmm_i(P)] } \right]
\end{equation}

These definitions enable us to characterize the Rawlsian inequality and rank-inefficiency of stable-dominating mechanisms, demonstrating that they may assign the worst-off student and the average student to schools ranked $\frac{n}{2}$ times worse than the respective first-best mechanisms.

\begin{proposition} \label{prop:rawlsian}
	For any mechanism $M^*$, any stable-dominating mechanism $M$ and all $n \in \mathbb{N}$:
	\begin{equation*}
		\poi(M;n) = \frac{n}{2} \geq	\poi(M^*;n)  .
	\end{equation*}
	Similarly, for any Pareto-efficient mechanism $M^\dagger$, any stable-dominating \emph{and} Pareto-efficient mechanism $M'$ and all $n \in \mathbb{N}$:
	\begin{equation*}
		\poik(M';n) = \frac{n}{2} \geq 	\poik(M^\dagger;n).
	\end{equation*}
\end{proposition}

\begin{proof}
	First we show that the $n/2$ ratio for Rawlsian inequality and rank-inefficiency is attainable by a stable-dominating mechanism. The school choice problem in the motivating example shows that the $n/2$ ratio can be achieved by the unique stable matching when $n=m=6$, and it can be generalized for any other $n=m$  values with the following construction. Let students' preferences be as follows:
	\begin{itemize}
		\item for $i_1$: $\rk_{i_1}(s_1)=1$ and $\rk_{i_1}(s_n)=2$.
		
		\item for any other $i_k$: $\rk_{i_k}(s_{k-1})=1$, $\rk_{i_k}(s_{k})=k$ and for any other schools with index $j \in [1:k-2]$, $\rk_{i_k}(s_j)=j+1$. 
	\end{itemize}
	
	Additionally, schools' priorities are so that for any $s_k$ and any two students $i_\ell, i_f$ with $k \leq \ell < f \leq n$, $\rk_{s_k}(i_\ell)< \rk_{s_k}(i_f)$.\footnote{Although we have not specified preferences and priorities fully, the given partial order will suffice for our purposes because any unspecified preferences or priorities will be irrelevant.}
	
	In DA's allocation, student $i_k$ is matched with $s_k$, for all $1 \leq k \leq n$, because $s_1$ is the best school for student $i_1$ and vice versa, and by stability of DA they must be matched together. Among the remaining schools and students, $i_2$ and $s_2$ rank each other as best, and must be matched, and so on. However, DA's allocation is Pareto-efficient and thus coincides with any stable-dominating mechanism. In any such mechanism $M$, $\rk_{i_n}[M_{i_n}(P)]=n$, even though the alternative Pareto-efficient matching $\mu_{i_j}=s_{j-1}$ (modulo $n$) is such that $\max_{i \in I} \rk_i(\mu_i)=2$, thus achieving the $n/2$ ratio for Rawlsian inequality. The same ratio is obtained for the rank inefficiency of the same example, as the sum of ranks in DA is simply $\sum_{i=1}^n i=\frac{n(n+1)}{2}$, whereas the alternative efficient allocation where $\mu_{i_j} = s_{j-1}$ (modulo $n$) achieves the following ranks: student $i_1$ receives rank 2 (assigned to $s_n$, his second choice), and students $i_2$ through $i_n$ each receive rank 1 (each assigned to their top choice). This yields a total rank sum of $2 + (n-1) = n + 1$, leading to the ratio $\frac{n(n+1)/2}{n+1} = \frac{n}{2}$.

	Now we show that the $n/2$ is a tight upper bound and that $\frac{n}{2}\geq \poi(M^*;n)$ for any Pareto-efficient mechanism $M^*$. For Rawlsian inequality: the maximum rank for a student is $n$, so the numerator cannot exceed $n$. The denominator can only be smaller than $2$ if $\max_{i \in I} \rk_i[\rw_i(P)]=1$. However, if such an allocation were possible, DA would stop at the first proposal round, assigning every student to their top school, in which case $\rw(P)=\da(P)$, leading to a Rawlsian inefficiency ratio of 1.
	
	For rank-inefficiency: any Pareto-efficient matching $M^\dagger$ can be obtained as a serial dictatorship, and in any serial dictatorship the $k$-th dictator cannot be assigned to a school with rank larger than $k$. Thus, in any serial dictatorship, the maximum rank is bounded from above by $\frac{(n+1)n}{2}$. The sum of ranks in $\rmm(P)$ cannot be any smaller than $n+1$, as otherwise $M^\dagger(P)=\rmm(P)$, and $\frac{(n+1)n}{2(n+1)}=\frac{n}{2}$ is a tight upper bound for the rank inefficiency of Pareto-efficient mechanisms, including those that are also stable-dominating.
\end{proof}

\subsection{Marginalized Students and Segregation}
\label{sec:marginalized}

Now we turn to the question of how stable-dominating mechanisms may improve (or not) the placement of marginalized students in society. Marginalized students can be thought of as those from low-income households, immigrants, ethnic, racial, or religious minorities, students with disabilities, and others facing systemic disadvantages that are engraved in schools' priorities. With this in mind, we are also interested in whether mechanisms that Pareto-dominate DA can alter the segregation in school compositions. Several studies have shown that DA generates high levels of school segregation \cite{terrier2021immediate, chileseg}, and therefore, understanding whether any stable-dominating mechanism can reduce the segregation created by DA is potentially interesting for policy questions.

We partition the set of students and schools into two non-empty groups. Some students will be \emph{advantaged} and some will be \emph{marginalized}. A marginalized student has a lower priority than every advantaged student at every school. The set of advantaged and marginalized students are denoted by $Y \subsetneq I$ and $Z \subsetneq I$, respectively, with $Y \cup Z = I$ and $Y \cap Z = \emptyset$. 
In this setup, we obtain the following result.
\begin{proposition}
	\label{prop:new}
	There is at least one marginalized student whose DA assignment does not improve under any stable-dominating mechanism.
\end{proposition}

\begin{proof}
Let $t$ be the last round of DA's execution where a marginalized student applies to a school. Such student $i$ is unimprovable if he is unmatched in DA by Lemma \ref{lemma:unim}. Thus, assume instead that $i$ is assigned to school $s$. 
Let $O^t(s)$ denote the set of student proposals that school $s$ receives during round $t$, and $q_s^{t-1}$ the number of students accepted by school $s$ at time $t-1$. 
If $|O^t(s)| +   q_s^{t-1} > q_s$, then school $s$ rejects at least one student at time $t$. Since $i$ is accepted and is a marginalized student, the student who was rejected at time $t$ was himself marginalized.
This rejected marginalized student either (1) applies to some school at a round $t'>t$, contradicting our definition of $t$, or (2) becomes unmatched, in which case he is unimprovable by Lemma \ref{lemma:unim}.
Therefore, $|O^t(s)| +   q_s^{t-1} \leq q_s$. Then school $s$ does not reject any student at time $t'\leq t$. But then $i$ is unimprovable by Lemma \ref{lemma:unim}.
\end{proof}

The existence of at least one unimprovable marginalized student raises a natural question: can marginalized students improve their assignments by trading with advantaged students? The following result shows this is impossible due to the priority structure.

\begin{proposition}
	\label{prop:nomixing}
	No trading cycle in $G^\da$ includes both advantaged and marginalized students.
\end{proposition}

\begin{proof}
	Suppose, for contradiction, that there exists a cycle containing both marginalized and advantaged students. Without loss of generality, there exists $i_2 \to i_1$ as part of such a cycle where $i_1$ is marginalized and $i_2$ is advantaged.
The edge $i_2 \to i_1$ means that advantaged student $i_2$ envies marginalized student $i_1$'s assignment, i.e., $i_2$ prefers school $s = \da_{i_1}(P)$ to their own assignment $\da_{i_2}(P)$.
However, by the definition of marginalized students, $i_1$ has lower priority than $i_2$ at every school, including school $s$. Then, $(i_2,s)$ is a blocking pair of $\da(P)$, contradicting stability.
Therefore, cycles in $G^{\da}(P)$ cannot contain both marginalized and advantaged students.
\end{proof}

Proposition \ref{prop:nomixing} reveals that improvement cycles respect the advantaged-marginalized divide—students can only trade within their own group. This segregation in trading possibilities has immediate implications for school compositions. To formalize these implications, we introduce the following definition.
We say that a school $s$ is \emph{fully segregated} under matching $\mu$ if it only accepts either advantaged or marginalized students, i.e. if $\mu^{-1}_s\subseteq Y$ or $\mu^{-1}_s\subseteq Z$. While this binary measure captures only the most extreme form of segregation, our results show that stable-dominating mechanisms cannot eliminate even these most severe cases of separation.

\begin{proposition}
	\label{prop:advantagedunim}
	Every advantaged student assigned to a school that is not fully segregated is unimprovable.
\end{proposition}

\begin{proof}
	Let $i$ be an advantaged student assigned to school $s$ under DA, where school $s$ also admits at least one marginalized student $j$.
	
	Suppose, for contradiction, that $i$ belongs to a cycle in $G^\da(P)$. By Proposition \ref{prop:nomixing}, this cycle contains only advantaged students. Therefore, there must exist an advantaged student $i'$ who envies $i$'s placement at school $s$.
	
	However, this is impossible. Since $s$ admits marginalized student $j$, and marginalized students have lower priority than all advantaged students, any advantaged student who applied to $s$ would have been accepted over $j$. The fact that $j$ is at $s$ means that no advantaged student other than those already at $s$ ever applied there.
	
	In particular, advantaged student $i'$ never applied to $s$ during DA's execution, which means $i'$ does not prefer $s$ to their current assignment. This contradicts our assumption that $i'$ envies $i$'s placement at $s$.
	
	Therefore, no advantaged student envies $i$, so $i$ cannot be part of any cycle and is unimprovable.
\end{proof}

We have established three key facts: some marginalized students are unimprovable (Proposition \ref{prop:new}), no trades occur across groups (Proposition \ref{prop:nomixing}), and advantaged students at mixed schools cannot improve (Proposition \ref{prop:advantagedunim}). These results combine to yield our main segregation Theorem.

\begin{theorem}[The Preservation of Segregation under Stable-Dominating Mechanisms]
	\label{thm:preservation}
	The number of advantaged and marginalized students accepted to each school under any stable-dominating mechanism is constant. Thus, if a school is fully segregated under DA, it remains fully segregated under any stable-dominating mechanism.
\end{theorem}

\begin{proof}
	Proposition \ref{prop:nomixing} establishes that marginalized and advantaged students cannot participate in the same trading cycles. Consider any school $s$ that admits only advantaged students under DA. By Proposition \ref{prop:nomixing}, no marginalized student can participate in any cycle that would give them access to school $s$, since any such cycle would necessarily involve trading with the advantaged students currently at $s$.
	
	Similarly, consider any school that admits only marginalized students under DA. No advantaged student can access this school through any cycle (nor would like to), since cycles cannot mix the two groups.
	
	For mixed schools, Proposition \ref{prop:advantagedunim} shows that advantaged students at such schools are unimprovable, while Proposition \ref{prop:new} ensures that marginalized students cannot trade into positions held by advantaged students at other schools.
	
	Therefore, the demographic composition of every school is perfectly preserved under any stable-dominating mechanism.
\end{proof}

Theorem \ref{thm:preservation} can be understood through an intuitive two-stage perspective. Any stable-dominating mechanism operates as if it first runs DA on advantaged students until a stable matching is formed, then allows marginalized students to apply to remaining seats. Since marginalized students cannot displace advantaged students due to priority structures, applications to schools filled by advantaged students would be automatically rejected. The occupancy invariance property (Lemma \ref{lemma:lema3}) thus applies sequentially: first preserving advantaged student occupancy at each school, then preserving marginalized student occupancy in the remaining capacity.

This two-stage structure reveals why efficiency improvements through stable-dominating mechanisms operate strictly within the constraints of the existing segregation pattern: advantaged students can only trade with other advantaged students, and marginalized students can only trade with other marginalized students. Theorem \ref{thm:preservation} extends naturally to settings with multiple tiers of students, where each tier has universally higher priority than lower tiers. In such cases, students can only trade within their own tier, preserving the exact tier composition at every school. 

This complete rigidity reveals an important limitation of efficiency-based reforms in school choice. While mechanisms like EADA may improve welfare within each group, they cannot bridge divides between groups or create more integrated school environments. Our results imply that the segregation patterns documented under DA in England \citep{terrier2021immediate} and Chile \citep{chileseg} would persist under any stable-dominating mechanism. Similarly, the limitations on disadvantaged students' access to good schools identified by \cite{calsamiglia2023catchment} cannot be addressed through Pareto-efficiency adjustments alone. Any policy aimed at reducing school segregation must therefore rely on interventions beyond Pareto improvements over stable matchings, such as changes to priority structures or quota systems.

To exemplify the preservation of segregation, consider the following school choice problem with eight students and four schools, each with two seats.
	\begin{table}[h!]
	\centering
	\caption{\centering DA allocation in squares, stable-dominating in bold.}
	\label{tab:example1}
	\begin{tabular}{CCCCCCCC|CCCC}
		\toprule
		i_1 & i_2  & i_3 & i_4 & i_5 & i_6 & i_7 & i_8 	& s_1 & s_2 & s_3 & s_4\\
		\midrule
		\pmb{s_2} & \pmb{s_2}  & \pmb{s_1} & \pmb{s_1} & s_1 & s_1 & \ensquare{$\pmb{s_3} $} 	& s_3		& i_1 & i_3& i_1 & i_1\\
		\ensquare{$s_1$} & \ensquare{$s_1$}  & \ensquare{$s_2$} & \ensquare{$s_2$} & s_2 & s_2 & s_4	& \ensquare{$\pmb{s_4}$} 		& i_2 & i_4 & i_2 & i_2 \\
		s_3 & s_3  & s_4 & s_3 & \ensquare{$\pmb{s_3}$} & s_3 & s_2 	& s_1			& i_5 & i_5 & i_3 & i_3 \\
		s_4 & s_4  & s_3 & s_4 & s_4 & \ensquare{$\pmb{s_4}$} & s_1 	& s_2													& i_3 & i_1 & i_4 & i_4 \\
&&&&&&&													& i_4 & i_2 & i_5 & i_5 \\
&&&&&&&													& i_8 & i_6	    & i_7 &  \\
&&&&&&&&i_6&& i_8&\\
&&&&&&&&&& i_6
	\end{tabular}
\end{table}

In the example above, students $\{i_1, \ldots, i_5\}$ are advantaged, whereas $\{i_6, i_7, i_8\}$ are marginalized. 
The DA matching assigns $\{i_1,i_2\}$ to $s_1$, $\{i_3, i_4\}$ to $s_2$, $\{i_5, i_7\}$ to $s_3$, and ${i_6, i_8}$ to $s_4$. 
Three out of four schools are fully segregated, admitting either only advantaged or marginalized students. 
Yet, the DA allocation is not Pareto-efficient: it is dominated by the following allocation generated by any Pareto-efficient and stable-dominating mechanism: $\{i_3,i_4\}$ to $s_1$, $\{i_1, i_2\}$ to $s_2$, $\{i_5, i_7\}$ to $s_3$, and $\{i_6, i_8\}$ to $s_4$. 
Note the preservation of segregation: the three schools that were fully segregated under DA remain segregated under the stable-dominating matching. 

Interestingly, there is an alternative Pareto-efficient allocation with less segregation, namely: $\{i_4,i_6\}$ to $s_1$, $\{i_1, i_2\}$ to $s_2$, $\{i_5, i_7\}$ to $s_3$, and $\{i_3, i_8\}$ to $s_4$. In this allocation, only one school $s_2$ is fully segregated, and furthermore, the rank-efficiency is better than in the stable-dominating outcome (average rank of 1.625 versus 1.75) as well as the Rawlsian inequality (of 3 versus 4 in the stable-dominating allocation). These three advantages come at a cost of more blocking pairs (4 versus only 2 in the stable-dominating allocation, namely ($i_3,s_1$), ($i_3,s_2$), ($i_5,s_1$), ($i_5,s_2$) versus ($i_5,s_1$),($i_5,s_2$)).

Pareto-efficient allocations with reduced segregation can be achieved through DA modifications that incorporate minority reserves \citep[e.g.][]{hafalir2013effective}.\footnote{There is a large literature studying minority reserves in school choice, which typically leads to matchings that do not Pareto dominate DA \citep[e.g.][]{kojima2012school,aygun2021college}. Our focus in this paper is specifically on mechanisms that maintain the Pareto-dominance property over DA, examining what such efficiency-preserving reforms can and cannot achieve.}
Empirical evidence demonstrates that such interventions can reduce segregation without harming efficiency: simulations of Swedish school choice show that reserved seats reduce segregation compared to proximity- or lottery-based priorities while maintaining high shares of students assigned to preferred schools \citep{kessel2018school}, while in Chile such reserves would reduce segregation while improving rank-efficiency \citep{escobar2022segregation}.

\section{Conclusion}
We have shown that stable-dominating mechanisms cannot eliminate three specific problems of DA: its rank-inefficiency, its inequality, and its associated segregation.
Our findings should be interpreted as a cautionary tale about efficiency-adjustments rather than arguing against its need and scope. Indeed, a substantial literature---in particular our companion paper \citep{ortega2025}---demonstrates that most students can benefit from mechanisms that Pareto-dominate DA. The contribution here is to precisely characterize what such mechanisms cannot fix, and the need for additional policies orthogonal to Pareto-efficiency to address such shortcomings.

\singlespacing      
\setlength{\parskip}{-0.3em} 
\bibliographystyle{te} 
\bibliography{bibliogr}

\end{document}